\documentclass[aps,prl,twocolumn,superscriptaddress]{revtex4}

\usepackage{color}
\usepackage{amsmath}
\usepackage{amssymb}
\usepackage{bbm}
\usepackage{quantum}
\usepackage{hyperref}
\usepackage[thmmarks,amsmath,hyperref]{ntheorem}
\usepackage[ntheorem,hyperref]{cleveref}

\renewcommand{\section}[1]{\bigskip\noindent{\bfseries #1.}}
\renewcommand{\subsection}[1]{\bigskip\noindent{\bfseries #1.}}
\renewcommand{\subsubsection}[1]{\bigskip\noindent{\bfseries #1.}}

\renewcommand{\cref}{\Cref}

\newtheorem{theorem}{Theorem}
\newtheorem{lemma}[theorem]{Lemma}

\newtheorem{proposition}[theorem]{Proposition}
\newtheorem{definition}[theorem]{Definition}

\theoremstyle{nonumberplain}
\theorembodyfont{\normalfont}
\theoremsymbol{$\Box$}
\newtheorem{proof}{Proof}

\begin{document}

\title{On the dimension of subspaces with bounded Schmidt rank}

\author{Toby Cubitt}
\affiliation{Department of Mathematics, University of Bristol,
  University Walk, Bristol BS8 1TW, UK}

\author{Ashley Montanaro}
\affiliation{Department of Computer Science, University of Bristol,
  Woodland Road, Bristol, BS8 1UB, UK}

\author{Andreas Winter}
\affiliation{Department of Mathematics, University of Bristol,
  University Walk, Bristol BS8 1TW, UK}
\affiliation{Quantum Information Technology Lab, National University
  of Singapore, 2 Science Drive 3, Singapore 117542}

\date{30 May 2007}

\begin{abstract}
  We consider the question of how large a subspace of a given
  bipartite quantum system can be when the subspace contains only
  highly entangled states. This is motivated in part by results of
  Hayden \emph{et al.}, which show that in large $d\times
  d$--dimensional systems there exist random subspaces of dimension
  almost $d^2$, all of whose states have entropy of entanglement at
  least $\log d - O(1)$. It is also related to results due to
  Parthasarathy on the dimension of completely entangled subspaces,
  which have connections with the construction of unextendible product
  bases. Here we take as entanglement measure the Schmidt rank, and
  determine, for every pair of local dimensions $d_A$ and $d_B$, and
  every $r$, the largest dimension of a subspace consisting only of
  entangled states of Schmidt rank $r$ or larger. This exact answer is
  a significant improvement on the best bounds that can be obtained
  using random subspace techniques. We also determine the converse:
  the largest dimension of a subspace with an \emph{upper} bound on
  the Schmidt rank. Finally, we discuss the question of subspaces
  containing only states with Schmidt \emph{equal to} $r$.
\end{abstract}

\maketitle

\section{Introduction}
Entanglement is at the heart of quantum information theory, and this
property of quantum systems is ultimately responsible for new
information tasks such as teleportation~\cite{teleport}, quantum key
agreement~\cite{qkd1,qkd2} or quantum computational
speedup~\cite{Shor}.  Consequently, a theory of measuring and
comparing the entanglement content of quantum states has
emerged~\cite{entanglement-survey}, which attempts to classify states
according to their non-classical capabilities. It is, however,
remarkable how large a number of entanglement measures have been put
forward~\cite{entanglement-survey,Christandl}, indicating that the
structure of entanglement is not one that can be captured by a single
number. One particular measure is the Schmidt rank of a pure bipartite
state $\ket{\psi}$, i.e.\ the number of non-zero coefficients
$\lambda_i$ in the --- essentially unique --- Schmidt form:
\begin{equation*}
  \ket{\psi} = \sum_i \lambda_i {\ket{e_i}}_A {\ket{f_i}}_B.
\end{equation*}
This measure has even been extended to mixed states, as the maximum
Schmidt rank in an optimal pure state
decomposition~\cite{TerhalHorodecki}, but the convex hull construction
could also be considered. For pure states, the Schmidt rank is indeed
the unique invariant under the class of stochastic local operations
and classical communication (SLOCC).

Here we ask, and answer completely, the question: what is the maximum
dimension of a subspace $S$ in a $d_A \times d_B$ bipartite system
such that every state in $S$ has Schmidt rank at least $r$? This is
trivial for $r=1$, so we can assume $r\geq 2$; also, the Schmidt rank
can be at most $\min\{d_A,d_B\}$, which we will assume without loss of
generality to be $d_A$.

There are two extreme cases. The first, $r=2$ (i.e.\ a subspace that
contains no product state), is addressed in~\cite{Parthasarathy}; the
answer is $(d_A-1)(d_B-1)$. The other, $r = d_A = d_B =: d$, has an
elementary solution: the answer is $1$ (take any one-dimensional
subspace spanned by a vector of maximum Schmidt rank $d$). To show
this, consider any two-dimensional subspace spanned by unit vectors
$\ket{\varphi}, \ket{\psi} \in \CC^d \otimes \CC^d$. We want to show
that at least one superposition $\ket{\phi_x} = \ket{\varphi} + x
\ket{\psi}$ has Schmidt rank less than $d$. The crucial observation is
that we can arrange the coefficients of a state vector $\ket{\phi}$ in
the computational basis $\{ \ket{i}\ket{j} \}_{i,j=1,\ldots,d}$, into
a $d \times d$ matrix $M(\phi)$, and that the Schmidt rank of the
state vector equals the linear rank of the associated matrix. In other
words, the statement that $\ket{\phi_x}$ has Schmidt rank less than
$r$ is captured by the vanishing of the determinant $\det
M(\phi_x)$. But the latter is a non-constant polynomial in $x$ of
degree $d$. Hence, it must have a root in the complex field, and the
corresponding $\ket{\phi_x}$ has Schmidt rank $r-1$ or less.

It turns out that the generalisation to arbitrary $r$ rests on the
same matrix representation, and the characterisation of Schmidt rank
via vanishing of certain determinants again plays a crucial role. It
involves, however, much deeper algebraic geometry machinery, extending
the above use of the fundamental theorem of algebra.

\section{Notation and Terminology}
We will denote the Schmidt rank of a bipartite pure state
$\ket[AB]{\psi}$ (nonzero, but not necessarily normalised) by
$\schmidt(\ket[AB]{\psi})$. We say that a subspace $S$ of the
bipartite space in question has Schmidt rank $\geq r$ if all its
nonzero vectors have Schmidt rank $\geq r$ (analogously for $\leq r$
and $=r$). If $M$ is a matrix, $R$ is a set of indices for rows of
$M$, and $C$ is a set of indices for columns, then $M_{\{R,C\}}$
denotes the submatrix formed by deleting all rows and columns
\emph{other than} those in $R$ and $C$.

The projective space of dimension $d$ is denoted by $\PP^d$.  This is
the space of lines (one-dimensional subspaces) of $\CC^{d+1}$; i.e.\
it is obtained from the nonzero elements of $\CC^{d+1}$ by identifying
collinear vectors. If a subset $S \subseteq \CC^{d+1}$ is a union of
lines, this identification associates to it a natural projectification
of $S$, denoted by $\PP(S) \subseteq \PP^d$; see
e.g.~\cite{Griffiths,Harris,Shafarevich,CoxLittleOShea} for this and
related notions from algebraic geometry.

For a state vector $\ket{\psi} \in \CC^{d_A} \otimes \CC^{d_B}$, with
fixed local bases of the two Hilbert spaces, such that
$\ket{\psi}=\sum_{i,j} c_{ij}\ket{i}\ket{j}$, define the $d_A \times
d_B$ matrix $M(\psi) = (c_{ij})_{i=1,\ldots,d_A,j=1,\dots,d_B}$. This
identifies $\CC^{d_A} \otimes \CC^{d_B}$ with the space
$\mathbb{M}(d_A,d_B)$ of $d_A \times d_B$ matrices.

\section{Bounding the dimension of highly entangled subspaces}
We first state some preparatory lemmas relating bipartite states to
matrices, whose proofs are widely known and do not warrant repetition
here.

\begin{lemma}
  \label{lem:rank_r_states}
  The set of (unnormalised) states in $\CC^{d_A}\otimes\CC^{d_B}$ with
  Schmidt rank $r$ is isomorphic to the set of $d_A\times d_B$ complex
  matrices with rank $r$.
\end{lemma}
\begin{proof}
  Obvious from the standard proof of the Schmidt decomposition via the
  singular value decomposition.
  % Writing the state $\ket{\psi}=\sum_{i,j} c_{ij}\ket{i}\ket{j}$ in
  % a product basis $\{\ket{i}\ket{j}\}_{i=1\dots d_A,j=1\dots d_B}$,
  % the coefficients $c_{ij}$ uniquely define a $d_A\times d_B$ matrix
  % $M(\psi)$ whose rank is equal to the Schmidt rank of the
  % state. Conversely, if $M$ is a $d_A\times d_B$ matrix and we fix a
  % product basis $\{\ket{i}\ket{j}\}_{i=1\dots d_A,j=1\dots d_B}$ for
  % $\CC^{d_A}\otimes\CC^{d_B}$, then the entries $M_{ij}$ uniquely
  % define a state $\ket{\psi} = \sum_{i,j} M_{ij}\ket{i}\ket{j}$
  % where $\schmidt(\ket{\psi}) = \rank M$.
\end{proof}

\begin{lemma}
  \label{lem:rank_r_minors}
  A matrix $M$ has $\rank M < r$ iff all its order--$r$ minors (the
  determinants of $r\times r$ submatrices) are zero.
\end{lemma}
\begin{proof}
  See \cite[p.13]{Horn+Johnson}.
\end{proof}

This means that a geometric characterisation of subspaces of Schmidt
rank $\geq r$ is to say that the linear space $M(S)$ of associated
matrices doesn't intersect the set of common zeros of all order--$r$
minors (except in the zero vector). Such common zeroes of sets of
multivariate polynomials are called (algebraic) varieties, and the one
in question has been studied in the mathematical
literature~\footnote{Determinantal varieties have also been studied in
  quantum information theory by Chen~\cite{Chen}, in a different
  context.}.

\begin{definition}[Determinantal variety]
  \label{def:determinantal_variety}
  The affine determinantal variety $\mathcal{D}_r(d_A,d_B)$ over the
  (algebraically closed) field $\FF$ in the space $\FF^{d_Ad_B}$ is
  the variety defined by the vanishing of all order--$r$ minors of a
  $d_A\times d_B$ matrix, whose elements are considered as independent
  variables in $\FF$.
\end{definition}
(Of course, in quantum theory we are mostly interested in the case
$\FF=\CC$.)

The basic idea is now essentially parameter counting: if the dimension
of $S$ plus that of the variety $\mathcal{D}_r(d_A,d_B)$ is larger
than $d_A d_B$, then the polynomial equations defining the order--$r$
minors have roots in $M(S)$. To make this heuristic rigorous, we need
to go to the corresponding projective spaces: since the polynomials
defined by the minors of a matrix are homogeneous, a determinantal
variety can also be thought of as a projective variety
$\PP(\mathcal{D}_r(d_A,d_B))$ in the space $\PP^{d_Ad_B-1}$. The same
is true for the subspace $S$, so it also has a projectification
$\PP(S)$.

\begin{lemma}[Dimension of determinantal varieties]
  \label{lem:determinantal_variety_dimension}
  The dimension of an affine determinantal variety is given by $\dim
  \mathcal{D}_r(d_A,d_B) = d_Ad_B - (d_A - r + 1)(d_B - r + 1)$.
\end{lemma}
\begin{proof}
  See e.g.~\cite[Proposition~12.2, p.~151]{Harris}.
\end{proof}
The corresponding projective determinantal variety has, of course,
dimension one less: $\dim \PP(\mathcal{D}_r(d_A,d_B)) = d_Ad_B - (d_A
- r + 1)(d_B - r + 1) - 1$. Likewise, the dimension of $\PP(S)$ is
$\dim S -1$.

Now, for projective varieties, the parameter counting argument always
holds:
\begin{lemma}[Intersection of projective varieties]
  \label{lem:intersection}
  If $V$ and $W$ are projective varieties in $\PP^d$ such that $\dim V
  + \dim W \geq d$, then $V\cap W\neq\emptyset$.
\end{lemma}
\begin{proof}
  See e.g.~\cite[Theorem~6, p.~76]{Shafarevich}
  or~\cite[Exercise~11.38, p.~148]{Harris}.
\end{proof}

\begin{proposition}
  \label{prop:dimension_upper_bound}
  For any subspace $S\subseteq\CC^{d_A}\otimes\CC^{d_B}$ of dimension
  $\dim S > (d_A-r+1)(d_B-r+1)$, there exists at least one state in
  the subspace with Schmidt rank strictly less than $r$.
\end{proposition}
\begin{proof}
  The set of all (unnormalised) states in the bipartite space
  $\CC^{d_A}\otimes\CC^{d_B}$ forms a projective space
  $\PP^{d_Ad_B-1}$ over the complex field. From
  \cref{lem:rank_r_states,lem:rank_r_minors,def:determinantal_variety},
  the subset of those states with Schmidt rank less than $r$ then
  forms a projective determinantal variety
  $\PP(\mathcal{D}_r(d_A,d_B))$ in that space. The subspace $S$
  corresponds to the projective variety $\PP(S)$ (a projective linear
  subspace), which has dimension $\dim\PP(S) > (d_A-r+1)(d_B-r+1) - 1$
  by assumption.

  Making use of \cref{lem:determinantal_variety_dimension}, we have
  \begin{equation*}
    \begin{split}
    \dim\PP(\mathcal{D}_r(d_A,d_B)) + \dim\PP(S)
       &\geq d_Ad_B - 1\\
       &= \dim\PP^{d_Ad_B-1}.
     \end{split}
  \end{equation*}
  Thus by \cref{lem:intersection}, $\PP(S)$ and
  $\PP(\mathcal{D}_r(d_A,d_B))$ have a non-empty intersection, i.e.\
  the subspace $S$ contains at least one state with Schmidt rank less
  than $r$.
\end{proof}

\section{Construction of highly entangled subspaces}
We will now give an explicit construction of a subspace with bounded
Schmidt rank that saturates the bound of
\cref{prop:dimension_upper_bound}, based on totally non-singular
matrices. (Note that we can not simply take the complement of
$\PP(\mathcal{D}_r(d_A,d_B))$ in $\PP^{d_Ad_B-1}$, since it is by no
means clear that this is a projective linear variety, i.e.\ a
subspace.)

\begin{definition}[Totally non-singular matrix]
  A matrix is said to be totally non-singular if all of its minors are
  non-zero.
\end{definition}

\begin{lemma}
  There exist totally non-singular matrices of any dimension.
\end{lemma}
\begin{proof}
  The $n\times n$ Vandermonde matrix generated by $0 < \lambda_1 <
  \lambda_2 < \dots < \lambda_n$ is totally positive (i.e.\ all its
  minors are strictly positive, see~\cite{Fallat}), therefore is also
  totally non-singular. Alternatively, it is also clear that a generic
  complex matrix will be totally non-singular, as the vanishing of a
  minor defines a set of matrices of measure $0$.
\end{proof}

\begin{lemma}
  \label{lem:vector_zero_elements}
  Let $M$ be an $m\times m$ totally non-singular matrix, with $m \geq
  n$. Let $v$ be any linear combination of $n$ of the columns of
  $M$. Then $v$ contains at most $n-1$ zero elements.
\end{lemma}
\begin{proof}
  Assume for contradiction that there exists a linear combination of
  $n$ columns of $M$ containing $n$ or more zero elements. Let $R$ be
  the set of indices of $n$ of those zero elements and $C$ be the set
  of indices of the $n$ columns. Since there is a linear combination
  of the columns of $M$ such that the elements indexed by $R$ are all
  zero, the columns of the submatrix $M_{\{R,C\}}$ are linearly
  dependent, thus the minor $\det M_{\{R,C\}}$ is zero and we have a
  contradiction.
\end{proof}

The construction of the subspace is based on the sets of vectors
introduced in \cref{lem:vector_zero_elements}.
\begin{proposition}
  \label{prop:dimension_lower_bound}
  Every bipartite system $\CC^{d_A}\otimes\CC^{d_B}$ has a subspace
  $S$ of Schmidt rank $\geq r$, and of dimension $\dim S =
  (d_A-r+1)(d_B-r+1)$.
\end{proposition}
\begin{proof}
  Since the bipartite states with Schmidt rank bounded by $r$ are
  isomorphic to $d_A\times d_B$ matrices whose rank is at least $r$
  (\cref{lem:rank_r_states}), and a matrix has rank greater than or
  equal to $r$ iff at least one of its order--$r$ minors is non-zero
  (\cref{lem:rank_r_minors}), it is sufficient to construct a set of
  linearly independent matrices $S$ of cardinality $\abs{S} =
  (d_A-r+1)(d_B-r+1)$ such that any linear combination of them has
  \emph{at least one} non-zero order--$r$ minor, since these then
  define a basis for a subspace with the desired properties.
  
  Label the diagonals of a $d_A\times d_B$ matrix by integers $k$,
  with $k$ increasing from lower-left to upper-right, and denote the
  length of the $k^{\mathrm{th}}$ diagonal by $\abs{k}$. From
  \cref{lem:vector_zero_elements}, there exist sets of $t=\abs{k}-r+1$
  linearly independent vectors of length $\abs{k}$ such that any
  linear combination of them has at most $t-1$ zero elements, or
  conversely, has at least $\abs{k}-(t-1)=r$ non-zero elements.

  For each diagonal with length $\abs{k} \geq r$, construct a set of
  linearly independent matrices $S_k$ of cardinality
  $\abs{S_k}=\abs{k}-r+1$ by putting these vectors down the
  $k^{\mathrm{th}}$ diagonal. By construction, any linear combination
  of these will have at least $r$ non-zero elements down that
  diagonal. Since the determinant of the $r\times r$ submatrix with
  these $r$ non-zero elements down its main diagonal is clearly
  non-zero, any linear combination of matrices in $S_k$ has at least
  one non-zero order--$r$ minor, thus has rank at least $r$.
  
  Now define the set $S = \bigcup_k S_k$. Since matrices from
  different $S_k$ have elements down different diagonals, the matrices
  in $S$ are linearly independent. It remains to show that any linear
  combination of matrices from \emph{different} $S_k$ still has rank
  at least $r$. Let $M$ be a matrix given by some linear combination
  of matrices in $S$, and let $\kappa$ be the maximum $k$ for which
  the linear combination includes matrices from $S_k$. It is still
  true that the $\kappa^{\mathrm{th}}$ diagonal of $M$ must contain at
  least $r$ non-zero elements. As $\kappa$ labels the top-rightmost
  diagonal of $M$ that contains any non-zero elements, the $r\times r$
  submatrix of $M$ with those $r$ non-zero elements down its main
  diagonal is lower-triangular, so has non-zero determinant. Thus $M$
  has at least one non-zero order--$r$ minor, so has rank at least
  $r$.
  
  Assume for convenience that $d_B \geq d_A$. To determine the
  cardinality of $S$, i.e.\ the dimension of the subspace, note that a
  $d_A\times d_B$ matrix has $1+d_B-d_A$ diagonals of length $d_A$,
  and $2$ diagonals of each length less than $d_A$. Then the
  cardinality of $S$ is given by
  \begin{align*}
    \Abs{S}
      &= \sum_k\Abs{S_k} = \sum_k(\abs{k}-r+1)\\
      &= \bigl(1+d_B-d_A\bigr)\bigl(d_A-r+1\bigr) + 2\sum_{i=r}^{d_A-1}(i-r+1)\\
      &= (d_A-r+1)(d_B-r+1),
  \end{align*}
  which matches the claimed dimension of the subspace.
\end{proof}

\section{Subspaces with bounded Schmidt rank}
Putting together
\cref{prop:dimension_upper_bound,prop:dimension_lower_bound}, we
obtain our main result:
\begin{theorem}
  The maximum dimension of a subspace $S \subseteq \CC^{d_A}\otimes
  \CC^{d_B}$ of Schmidt rank $\geq r$ is given by
  $(d_A-r+1)(d_B-r+1)$.  \hfill $\Box$
\end{theorem}

One could instead ask for the converse: subspaces of Schmidt rank
$\leq r$. Note that geometrically this corresponds to a linear
subspace lying \emph{within} the determinantal variety
$\mathcal{D}_{r+1}(d_A,d_B)$. There is a simple construction, $S = R
\otimes \CC^{d_B}$, for any subspace $R \subseteq \CC^{d_A}$ of
dimension $r$, which achieves $\dim S = r d_B$. This is clearly tight
if $r=1$ or $r=d_A$. In fact, one can show that this construction is
optimal in general, which is immediate from the following theorem due
to Flanders~\cite{Flanders}:
\begin{theorem}[Flanders]
  Let $S$ be a subspace of the space of $d_A \times d_B$ matrices,
  where $d_A \leq d_B$. Let $r$ be the maximum rank of any element of
  $S$. Then $\dim S \leq r d_B$.  \hfill $\Box$
\end{theorem}

Another interesting variant is to ask what are the subspaces which
have Schmidt rank \emph{exactly} $r$. For example, our construction
above yields subspaces of dimension $d_B-d_A+1$ in $\CC^{d_A} \otimes
\CC^{d_B}$ of Schmidt rank equal to $d_A$. A different example is
given by the three-dimensional completely antisymmetric subspace of
$\CC^3\otimes \CC^3$, which has Schmidt rank equal to $2$. This
question has been the subject of a remarkably long-running study in
the linear algebra literature and, as far as we are aware, the general
case remains unsolved. The best existing results are summarised in the
following theorem, which can be found in~\cite{westwick}:
\begin{theorem}[Westwick]
  Let $S$ be the largest subspace of the space of $d_A \times d_B$
  matrices, with $d_B \geq d_A$, such that the rank of every non-zero
  element of $S$ is $r$. Then in general,
  \begin{equation*}
      d_B-r+1 \leq \dim S \leq d_A + d_B -2r + 1.
  \end{equation*}
  Furthermore, if $d_B - r + 1$ does not divide $(d_A-1)!/(r-1)!$,
  then $\dim S = d_B-r+1$. If $d_A=r+1, d_B=2r-1$, then $\dim S = r +
  1$.  \hfill $\Box$
\end{theorem}
For sufficiently large $d_B$, it is of course impossible for $d_B-r+1$
to divide $(d_A-1)!/(r-1)!$, so the result for that case applies to
all sufficiently high-dimensional spaces. It appears to be impossible
to obtain this result via general dimensional arguments similar to
those used in this paper, which only reproduce the general upper
bound, $\dim S \leq (d_B-r+1)+(d_A-r)$.

\section{Discussion: applications and open questions}
We have determined the exact maximum dimension of subspaces of Schmidt
rank $\geq r$ in any bipartite quantum system. The upper bound on the
dimension is a generalisation of Parthasarathy's
argument~\cite{Parthasarathy} for a subspace avoiding the manifold of
product states, to the avoiding of a determinantal variety. Our
constructive lower bound seems to differ from Parthasarathy's (in the
case $r=2$), which is based on unextendible product vector systems.

Comparing these results, using the Schmidt measure,
with~\cite{random_subspace}, where the entropy measure of entanglement
is used, we have much tighter control on the entanglement in
subspaces. For example, in the cited paper, the random subspaces that
are constructed are necessarily highly entangled, simply because that
is the generic behaviour of random states. In contrast, here we find
the largest subspaces of bounded Schmidt rank over the \emph{whole
  range} of the entanglement measure, including values far away from
typical. This is most clearly demonstrated by considering subspaces
with Schmidt rank within a constant fraction of the maximum: $r\geq
kd_A$. For $k\geq 2^{-d_A/(d_B\ln 2)}$, using the results
of~\cite[Theorem~IV.1]{random_subspace} gives nothing better than the
trivial one-dimensional subspace, yet the exact result is
asymptotically of order $(1-k)^2 d_Ad_B$, i.e.\ within a constant
fraction of the entire space!

Our results can be used, in the spirit of~\cite{random_subspace}, to
construct highly mixed states of very large Schmidt
measure~\cite{TerhalHorodecki}: let $\rho$ be the normalised projector
onto a maximum dimensional subspace $S$ of Schmidt rank $\geq
r$. Then, since every pure state decomposition of $\rho$ can only
consist of state vectors from $S$, any entanglement measure built from
the Schmidt ranks of the constituent pure states has to be at least
$r$. For example, in arbitrarily large $d\times d$--systems, we thus
find for any $p$ states of rank $\geq p^2 d^2$ (i.e. entropy $2\log d
+ 2\log p$) and Schmidt measure $\geq (1-p)d$. The ideas and results
described in this paper also have applications to the study of
degradable quantum channels~\cite{toby}.

The present paper raises a number of questions: a first is about
multiparty generalisations, e.g.\ looking at subspaces with
constraints on the Schmidt rank across all bipartite cuts. The case of
completely entangled subspaces was solved in~\cite{Parthasarathy}.
Note that, \emph{were} random subspaces to saturate the bound in
\cref{prop:dimension_lower_bound}, a random subspace saturating the
tightest of the bipartite constraints would automatically satisfy all
the other constraints. The multipartite case would therefore reduce to
the bipartite case. However, this would contradict the known result
for completely entangled subspaces, once again underlining the fact
that progress in this type of problem requires going beyond the
typical case.

We could also attempt to make statements about more operationally
motivated entanglement measures, especially those based on von Neumann
or R\'enyi entropies, as for example
in~\cite{random_subspace}. However, the algebraic techniques used here
do not seem to give any insight into these problems.

\section{Acknowledgments}
The authors acknowledge support by the European Commission, project
``QAP'', and the U.K.\ EPSRC through postgraduate scholarships, the
``QIP IRC'' and an Advanced Research Fellowship.

We thank Aram Harrow, Richard Jozsa, Richard Low and Will Matthews for
various spirited discussions about the content of this paper;
especially the latter three for their solution of the case $r = d_A =
d_B$, which was the starting point of this work. We would also like to
thank Wee Kang Chua for drawing our attention to the question of
subspaces with upper-bounded Schmidt rank.

\bibliography{Schmidt_rank_subspace}

\end{document}